%% file: Main.tex
\documentclass[copyright,creativecommons]{eptcs}
\providecommand{\event}{ITRS 2010} % Name of the event you are submitting to
\usepackage{xcolor}
\usepackage{latexsym,amsmath,amssymb}
\usepackage{prooftree}
\usepackage{listings}
\usepackage{enumerate}
    \providecommand{\event}{Betti Venneri (Ed.): ITRS 2010} % adapt
    \providecommand{\firstpage}{1} % adapt
\input{Def}

\title{Intersection types for unbind and rebind\thanks{This work has been partially supported
by MIUR DISCO - Distribution, Interaction, Specification,
Composition for Object Systems, and IPODS - Interacting Processes in Open-ended Distributed Systems.}}
\author{Mariangiola Dezani-Ciancaglini
\institute{Dip. di Informatica, Univ. di Torino, Italy}
\and Paola Giannini
\institute{Dip. di Informatica, Univ. del Piemonte Orientale, Italy}
 \and Elena Zucca
 \institute{DISI, Univ. di Genova, Italy}
 }
\def\titlerunning{Intersection types for unbind and rebind}
\def\authorrunning{Dezani, Giannini and Zucca}

\begin{document}
\maketitle

\begin{abstract}
We define a type system with intersection types for an extension of
lambda-calculus with unbind and rebind operators. In this calculus,
a term $\te$ with free variables $\x_1,\ldots,\x_n$, representing
open code, can be packed into an \emph{unbound} term
$\Unbind{\x_1,\ldots,\x_n}{\te}$, and passed around as a value. In
order to execute inside code, an unbound term should be explicitly
\emph{rebound} at the point where it is used. Unbinding and
rebinding are hierarchical, that is, the term $\te$ can contain
arbitrarily nested unbound terms, whose inside code can only be
executed after a sequence of rebinds has been applied.
Correspondingly, types are decorated with levels, and a term has
type $\TypeWithLevel{\tau}{\level}$ if it needs $\level$ rebinds in
order to reduce to a value of type $\tau$. With intersection types
we model the fact that a term can be used differently in contexts
providing different numbers of unbinds. In particular, top-level
terms, that is, terms not requiring unbinds to reduce to values,
should have a \emph{value} type, that is, an intersection type where
at least one element has level 0. With
the proposed intersection type system we get soundness under the
call-by-value strategy, an issue which was not resolved by previous
type systems.
\end{abstract}

\input{Intro}
\input{Calculus}
\input{TypeSystem}
\input{Soundness}
\input{Conclu}
\bibliographystyle{eptcs}
\bibliography{Main}

\end{document}

%% file: Def.tex
\newif\ifsubmit
\submittrue
\ifsubmit
\newcommand{\EZ}[1]{#1}
\newcommand{\EZComm}[1]{}
\newcommand{\MD}[1]{#1}
\newcommand{\MDComm}[1]{}
\newcommand{\PG}[1]{#1}
\newcommand{\PGComm}[1]{}

\newcommand{\pgComm}[1]{}
\else
\newcommand{\EZ}[1]{\textcolor{blue}{#1}}
\newcommand{\EZComm}[1]{{\scriptsize \textcolor{blue}{[Elena{:} #1]}}}
\newcommand{\MD}[1]{\textcolor{red}{#1}}
\newcommand{\MDComm}[1]{{\scriptsize \textcolor{red}{[Mariangiola{:} #1]}}}
\newcommand{\PG}[1]{\textcolor{magenta}{#1}}
\newcommand{\PGComm}[1]{{\scriptsize \textcolor{magenta}{[Paola{:} #1]}}}

\newcommand{\pgComm}[1]{{\scriptsize \textcolor{green}{[Paola{:} #1]}}}\fi

\newtheorem{theorem}{Theorem}
\newtheorem{lemma}[theorem]{Lemma}
\newenvironment{proof}[1][Proof]{\begin{trivlist}
\item[\hskip \labelsep {\bfseries #1}]}{\end{trivlist}}

\newcommand{\HSep}{\hbox to \textwidth{\bf\hrulefill}}
\newcommand{\lab}[1]{\ensuremath{\scriptstyle{\textsc{(#1)}}}}
\newcommand{\Space}{\hskip 0.8em}
\newcommand{\NamedRule}[4]{\scriptstyle{\textsc{(#1)}}\
\displaystyle                  %  #1 = nome regola
\frac{#2}{#3}\           %  #2 = premesse (modo math) %  #3 = conseguenza (modo math)
\scriptstyle{#4}    %  #4 = side conditions (modo math)
}

\newcommand{\Rule}[3]{{\scriptstyle{}}\
\displaystyle
\frac{#1}{#2}\           %  #1 = premesse (modo math) %  #2 = conseguenza (modo math)
#3     %  #3 = side conditions (modo math)
}

\newcommand{\refToFigure}[1]{Figure~\ref{fig:#1}}
\newcommand{\refToSection}[1]{Section~\ref{sect:#1}}

\newenvironment{grammatica}{$\begin{array}[t]{lcll}}{\end{array}$}
\newcommand{\produzione}[3]{#1&{:}{:}=&#2 & \mbox{{\small{#3}}}}
\newcommand{\terminale}[1]{{\text{\tt #1}}}

\newcommand{\te} {\mathit{t}} %metavariabile termini
\newcommand{\val}{\mathit{v}}%metavariabile valori
\newcommand{\x}{\mathit{x}} %(meta)variabile generica
\newcommand{\y}{\mathit{y}} %(meta)variabile generica
\newcommand{\z}{\mathit{z}} %(meta)variabile generica
\newcommand{\nval}{\mathit{n}}%metavariabile numerali
\newcommand{\tenv}{\Gamma}
\newcommand{\error}{\mathit{error}}
\newcommand{\BinExp}[3]{#2\mathrel{#1}#3}
\newcommand{\SumExp}[2]{\BinExp{\SumKw}{#1}{#2}}
\newcommand{\SumKw}{\terminale{+}}
\newcommand{\LambdaExp}[2]{\lambda #1.#2}
\newcommand{\AppExp}[2]{#1\, #2}
\newcommand{\Unbind}[2]{\langle\;#1{\;\mid\;}#2\;\rangle}
\newcommand{\Rebind}[2]{#1[#2]}
\newcommand{\tsubst}{\mathit{r}}
\newcommand{\tsubstval}{\tsubst^\val}
\newcommand{\context}{{\cal E}}
\newcommand{\emptycontext}{[\,]}
\newcommand{\subst}{\sigma}

\newcommand{\ev}{\longrightarrow}
\newcommand{\evstar}{\ev^\star}
\newcommand{\extractTEnv}[1]{\mathit{tenv}(#1)}
\newcommand{\extractSubst}[1]{\mathit{subst}(#1)}
\newcommand{\FV}{\mathit{FV}}
\newcommand{\Int}{{\mathbb Z}}
\newcommand{\ApplySubst}[2]   {#1\{#2\}}
\newcommand{\RestrictSubst}[2]   {#1_{\mid #2}}
\newcommand{\CancelSubst}[2]   {#1_{\setminus #2}}
\newcommand{\dom}{\mathit{dom}}

\newcommand{\inContext}[1]{\context[#1]}

\newcommand{\level}{\mathit{k}} % livello
\newcommand{\Nat}{{\mathbb N}}
\newcommand{\levelPrime}{\mathit{h}} % livello
\newcommand{\intType}{\terminale{int}}
\newcommand{\codeType}{\terminale{code}}
\newcommand{\TypeWithLevel}[2]{#1^{#2}}
\newcommand{\T}{\mathit{T}}
\newcommand{\V}{\mathit{V}}
\newcommand{\funType}[2]{#1 \rightarrow {#2}}
\newcommand{\Inter}[2]{#1\wedge#2}

\newcommand{\incrLevel}[2]{(#1)^{\oplus#2}}
\newcommand{\congr}{\equiv}
\newcommand{\IsWFExp}[3]           {#1\vdash #2 : #3}
\newcommand{\IsWFSubst}[3]           {#1\vdash #2:#3}
\newcommand{\Subst}[2]   {{#1,#2}}
\newcommand{\OK}{\sf{ok}}

\newcommand{\DLet}[3] {\terminale{dlet} \mathrel{#1}
\terminale{=}\mathrel{#2}\terminale{in}\mathrel{#3}}

%% file: Intro.tex
\section*{Introduction}
In {previous work} \cite{DezaniEtAl09,DezaniEtAl10} we {introduced} an extension
of lambda-calculus with unbind and rebind operators, providing a
simple unifying foundation for dynamic scoping, rebinding and
delegation mechanisms. This extension {relies} on the following
ideas:
\begin{itemize}
\item A term $\Unbind{\tenv}{\te}$, where $\tenv$ is a set of
{typed variables} called \emph{unbinders}, is a value{, of a special type $\codeType$,} representing ``open
code'' which may contain free variables in the domain of $\tenv$.
\item To be used, open code should be \emph{rebound} through the
operator  $\Rebind{\te}{\tsubst}$, where $\tsubst$ is a (typed) substitution
(a {map} from {typed} variables to terms). Variables in the domain of
$\tsubst$ are called \emph{rebinders}. When the rebind operator is
applied to a term $\Unbind{\tenv}{\te}$, a dynamic check is
performed: if all unbinders are rebound with values of the required
types, then the substitution is performed, otherwise a dynamic error
is raised.
\end{itemize}

{For instance, the term\footnote{In the examples we omit type annotations when they are
irrelevant.}
 $\Rebind{\Unbind{\x,\y}{\SumExp{\x}{\y}}}{\x\mapsto
1,\y\mapsto 2}$ reduces to $\SumExp{1}{2}$, whereas both
$\Rebind{\Unbind{\x,\y}{\SumExp{\x}{\y}}}{\x\mapsto 1}$ and
$\Rebind{\Unbind{\x{:}\intType}{\SumExp{\x}{1}}}{\x{:}\funType{\intType}{\intType}\mapsto
\LambdaExp{\y}{\SumExp{\y}{1}}}$ reduce to
$\error$.}

{Unbinding and rebinding are hierarchical, that is, the term
$\te$ can contain arbitrarily nested unbound terms, whose inside
code can only be executed after a sequence of rebinds has been
applied\footnote{See the Conclusion for more comments on this
choice.}. For instance, \emph{two} rebinds must be applied to the
term $\Unbind{\x}{\x+\Unbind{\x}{\x}}$ in order to get an integer:
\begin{center} $
\begin{array}{rcl}
\Rebind{\Rebind{\Unbind{\x}{\x+\Unbind{\x}{\x}}}{\x\mapsto
1}}{\x\mapsto 2} & {\ev} & \Rebind{(1+\Unbind{\x}{\x})}{\x\mapsto
2}\\
& \ev& (\Rebind{1}{\x\mapsto 2})+(\Rebind{\Unbind{\x}{\x}}{\x\mapsto
2})
\\
& \evstar& 1+ 2
\end{array}$
\end{center}
}

{Correspondingly, types are decorated with levels, and a term has
type $\TypeWithLevel{\tau}{\level}$ if it needs $\level$ rebinds in
order to reduce to a value of type $\tau$. {With intersection
types we} model the fact that a term can be used differently in
contexts which  provide a different number $\level$ of unbinds. For
instance, the term $\Unbind{\x}{\x+\Unbind{\x}{\x}}$ above has type
$\Inter{\TypeWithLevel{\intType}{2}}{\TypeWithLevel{\codeType}{0}}$,
since it can be safely used in two ways: either in a context which
provides two rebinds, as shown above, or as a value of type
$\codeType$, as, e.g., in:
\begin{center}
$\begin{array}{l}
\AppExp{(\LambdaExp{\y}{\Rebind{\Rebind{\y}{\x\mapsto 1}}{\x\mapsto
2}})}{\Unbind{\x}{\x+\Unbind{\x}{\x}}}
\end{array}$
\end{center}
{On the other side, the term $\Unbind{\x}{\x+\Unbind{\x}{\x}}$ does
\emph{not} have type $\TypeWithLevel{\intType}{1}$, since by
applying only one rebind {with $[{\x\mapsto 1}]$} we get the term
$1+\Unbind{\x}{\x}$ which is stuck.}}

{The} {use of intersection
types allows us to get soundness w.r.t. the call-by-value strategy.
{This} issue was not resolved by {previous type systems}
\cite{DezaniEtAl09, DezaniEtAl10}} where, for this reason, we only
considered the call-by-name reduction strategy. To see the problem,
consider the following example.

The term
\begin{center}
$\begin{array}{l}
\AppExp{(\LambdaExp{\y} {\Rebind{\y}{\x\mapsto 2}})}{(1+\Unbind{\x}{\x})}
\end{array}$
\end{center}
is stuck in the call-by-value strategy, since the argument is not a
value, hence should be ill typed{, even though} the argument has type
$\TypeWithLevel{\intType}{1}$, which is a correct type for the
argument of the function. By using intersection types, this can be
{enforced} by requiring arguments of functions to have
\emph{value types}, that is, intersections where (at least) one of
the conjuncts is a type of level $0$. In this way, the above term
is ill typed. Note that a call-by-name evaluation of the above
term gives
\begin{center} $
\begin{array}{rcl}
\AppExp{(\LambdaExp{\y}{\Rebind{\y}{\x\mapsto
2}})}{(1+\Unbind{\x}{\x})}& \ev &
(1+\Unbind{\x}{\x}) {[{\x\mapsto 2}]}\\
&  \ev & (\Rebind{1}{\x\mapsto 2})+
(\Rebind{\Unbind{\x}{\x}}{\x\mapsto 2})\\
&  \evstar & 1+2.
\end{array}
$
\end{center}
Instead, the term $\AppExp{(\LambdaExp{\y}{\Rebind{\y}{\x\mapsto
2}})}{\Unbind{\x}{1+\x}}$ is well typed, and it reduces as follows in both call-by-value
and call-by-name strategies:
\begin{center} $
\begin{array}{rcl}
\AppExp{(\LambdaExp{\y}{\Rebind{\y}{\x\mapsto
2}})}{\Unbind{\x}{1+\x}}& \ev &
\Unbind{\x}{1+\x} {[{\x\mapsto 2}]}\\
&  \ev & 1+2.
\end{array}
$ \end{center}

 It is interesting to note that this phenomenon is due to
the possibility for operators (in our case for +) of acting on
arguments which are unbound terms. This design choice is quite
natural in view of discussing {open code, as in MetaML
\cite{MetaML}.} In pure $\lambda$-calculus there is no closed term
which converges when evaluated by the lazy call-by-name strategy and
is stuck when evaluated by the call-by-value strategy. Instead there
are closed terms, like $(\lambda x. \lambda y.y)((\lambda
z.z\,z)(\lambda z.z\,z))$, which converge when evaluated by the lazy
call-by-name strategy and diverge when evaluated by the
call-by-value strategy, and open terms, like $(\lambda x.\lambda
y.y)z$, which converge when evaluated by the lazy call-by-name
strategy and are stuck when evaluated by the call-by-value
strategy.\footnote{Note that following Pierce \cite{pierce} we
consider only $\lambda$-abstractions as values, while for Plotkin
\cite{plot77} also variables are values.}

{In summary, the contribution of this paper is the following. We} define a type system %with intersection types
for the calculus of Dezani et al. \cite{DezaniEtAl09,DezaniEtAl10}, where,
differently from those papers, we omit types on the lambda-binders
in order to get the whole expressivity of the intersection type
constructor \cite{Venn94}. The type system shows, in our opinion, an
interesting and novel application of intersection types. Indeed,
they handle in a uniform way the three following issues.
\begin{itemize}
\item {Functions may be} applied to arguments of (a finite set of) different types.
\item {A term can be used differently in contexts providing different numbers} of unbinds.
Indeed, an intersection type for a term includes a type of form $\TypeWithLevel{\tau}{\level}$ if the term needs $\level$ rebinds in order to reduce to a value of type $\tau$.
\item {Most notably, the type system guarantees soundness for the call-by-value
strategy, by requiring that top-level terms, that is, terms which do
not require unbinds to reduce to values}, should have {value
types.}
\end{itemize}

\paragraph{\bf Paper Structure.}In \refToSection{syntaxsem} we {introduce the syntax and the operational
semantics of the language}. In \refToSection{typesystem} we  define
the type system and prove its soundness in \refToSection{soundness}. In \refToSection{rwc} we
discuss related and further work.

%% file: Calculus.tex
\section{{Calculus}}\label{sect:syntaxsem}
The syntax and {reduction rules} of the calculus are given in
\refToFigure{calculus}.

\begin{figure}[h]
\HSep
\begin{grammatica}
\produzione{\te}{\x\mid\nval\mid\SumExp{\te_1}{\te_2}\mid\LambdaExp{\x}{\te}\mid\AppExp{\te_1}{\te_2}\mid\Unbind{\tenv}{\te}\mid\Rebind{\te}{\tsubst}\mid\error}{{term}}\\
\produzione{\tenv}{\x_1{:}\T_1,\ldots,\x_m{:}\T_m}{{type context}}\\
\produzione{\tsubst}{\x_1{:}\T_1\mapsto\te_1,\ldots,\x_m{:}\T_m\mapsto\te_m }{{(typed) substitution}}   \\ \\
\produzione{\val}{\LambdaExp{\x}{\te}\mid\Unbind{\tenv}{\te}\mid\nval}{{value}}\\
\produzione{\tsubstval}{\x_1{:}\T_1\mapsto\val_1,\ldots,\x_m{:}\T_m\mapsto\val_m }{{value substitution}} \\ \\
\produzione{\context}{\emptycontext\mid\SumExp{\context}{\te}\mid\SumExp{\nval}{\context}\mid\AppExp{\context}{\te}\mid\AppExp{\val}{\,\context}\mid\Rebind{\te}{\tsubst,\x{{:}\T}\mapsto
\context}}{{evaluation context}}\\
\produzione{{\subst}}{{\x_1\mapsto\val_1,\ldots,\x_m\mapsto\val_m}}{{(untyped) substitution}}
\end{grammatica}

\HSep

\begin{center}
$\begin{array}{l@{\quad}l@{\quad}l}
 {\SumExp{\nval_1}{\nval_2}\ev\nval} &
  {\mbox{if}\quad\tilde{\nval}=\tilde{\nval}_1+^\Int\tilde{\nval}_2} & \lab{Sum} \\
   {\AppExp{(\LambdaExp{\x}{\te})}{\val}\ev\ApplySubst{\te}{\x\mapsto\val}}{}{}
   & & \lab{App}
% \\
%     {\AppExp{\Unbind{\tenv}{\te}}{\val}\ev\Unbind{\tenv}{\AppExp{\te}{\val}}}
%   & {\mbox{if}\quad\FV(\val)\cap\dom(\tenv)=\emptyset} & \lab{AppUnbind}
 \\
  {\Rebind{\Unbind{\tenv}{\te}}{\tsubstval}\ev\ApplySubst{\te}{\RestrictSubst{\extractSubst{\tsubstval}}{\dom(\tenv)}}}
  & \mbox{if}\quad\tenv\subseteq\extractTEnv{\tsubstval} & \lab{{RebindUnbindYes}} \\
{\Rebind{\Unbind{\tenv}{\te}}{\tsubstval}\ev\error} &
{\mbox{if}\quad\tenv\not\subseteq\extractTEnv{\tsubstval}}{} & \lab{{RebindUnbindNo}} \\
 {\Rebind{\nval}{\tsubstval}\ev\nval}& & \lab{RebindNum}\\
 {\Rebind{(\SumExp{\te_1}{\te_2})}{\tsubstval}\ev\SumExp{\Rebind{\te_1}{\tsubstval}}{\Rebind{\te_2}{\tsubstval}}}
 & & \lab{RebindSum} \\
{\Rebind{(\LambdaExp{\x}{\te})}{\tsubstval}\ev\LambdaExp{\x}{\Rebind{\te}{\tsubstval}}}
 & & \lab{RebindAbs} \\
 {\Rebind{(\AppExp{\te_1}{\te_2})}{\tsubstval}\ev\AppExp{\Rebind{\te_1}{\tsubstval}}{\Rebind{\te_2}{\tsubstval}}}
 & & \lab{RebindApp} \\
  {\Rebind{\Rebind{\te}{\tsubst}}{\tsubstval}\ev\Rebind{\te'}{\tsubstval}}
 & \mbox{if}\quad\Rebind{\te}{\tsubst}\ev\te' & \lab{RebindRebind} \\
   {\Rebind{\error}{\tsubstval}\ev\error}& &\lab{RebindError} \\
\\
\end{array}$
$\begin{array}{l@{\quad\quad}l}
 \prooftree  \te\ev \te'\quad\quad \context\not=\emptycontext
  \justifies \inContext{\te}\ev\inContext{\te'} \using \lab{Ctx}
\endprooftree &
\prooftree \te\ev \error \quad\quad \context\not=\emptycontext
  \justifies \inContext{\te}\ev\error \using \lab{CtxError}
\endprooftree
\end{array}$
\end{center}
\caption{Syntax and reduction rules}\label{fig:calculus} \HSep
\end{figure}

Terms of the calculus are the $\lambda$-calculus terms, the unbind
and rebind constructs, and the dynamic error. {Moreover, we}
include integers with addition to show how unbind and rebind behave
on primitive data types. {Unbinders and rebinders are annotated
with types $\T$, which will be described in the following section.
Here it is enough to assume that they include standard $\intType$
and functional types.} Type contexts and
substitutions are assumed to be maps, that is, order is immaterial
and variables cannot appear twice.

Free variables and application of a substitution to a term are
defined in \refToFigure{subst}. {Note that an unbinder behaves
like a $\lambda$-binder:  for instance, in a term of shape
$\Unbind{\x}{\te}$, the unbinder $\x$ introduces a local scope, that
is, binds free occurrences of $\x$ in $\te$.  Hence, a substitution
for $\x$ is not propagated inside $\te$. Moreover, a condition,
 which prevents
capture of free variables similar to the $\lambda$-abstraction
case is needed. For instance, the
term
$\AppExp{(\LambdaExp{\y}{\Unbind{\x}{\y}})}{(\LambdaExp{\z}{\x})}$
reduces to $\ApplySubst{\Unbind{\x}{\y}}{\y \mapsto
\LambdaExp{\z}{\x}}$ which is stuck, i.e., it does not reduce to
$\Unbind{\x}{\LambdaExp{\z}{\x}}$,} which would be
wrong.

However, $\lambda$-binders and unbinders behave differently
w.r.t. $\alpha$-equivalence. A $\lambda$-binder can be renamed, as
usual, together with all its bound variable occurrences, whereas this
is \emph{not} safe for an unbinder: for instance,
$\Rebind{\Unbind{\x}{\SumExp{\x}{1}}}{\x\mapsto 2}$ is  not
equivalent to $\Rebind{\Unbind{\y}{\SumExp{\y}{1}}}{\x\mapsto 2}$.
Only a \emph{global} renaming, e.g., leading to
$\Rebind{\Unbind{\y}{\SumExp{\y}{1}}}{\y\mapsto 2}$, would be
safe.\footnote{A more sophisticated solution {\cite{BiermanEtAl03a}} allows local renaming of unbinders by a
``precompilation'' step annotating variables with \emph{indexes},
which can be $\alpha$-renamed, but are not taken into account by the
rebinding mechanism.  Indeed, variable occurrences which are
unbinders, rebinders, or bound to an unbinder,  actually play the
role of \emph{names} rather than standard variables. Note that 
variables in a rebinder, e.g., $\x$  in $[\x\mapsto 2]$, are not
bindings, and $\x$ is neither a free, nor a bound variable. See the
Conclusion for more comments on this difference.}

\begin{figure}[h]
\HSep

$\begin{array}{l}
\\
  \FV(\x)=\{\x\}\\
  \FV(\nval)=\emptyset\\
  \FV(\SumExp{\te_1}{\te_2})=\FV(\te_1)\cup\FV(\te_2)\\
  \FV(\LambdaExp{\x}{\te})=\FV(\te)\setminus\{\x\}\\
  \FV(\AppExp{\te_1}{\te_2})=\FV(\te_1)\cup\FV(\te_1)\\
  \FV(\Unbind{\tenv}{\te})=\FV(\te)\setminus\dom(\tenv)\\
 { \FV(\Rebind{\te}{\tsubst})=\FV(\te)\cup\FV(\extractSubst{\tsubst})}\\
  { \FV(\x_1\mapsto\te_1,\ldots,\x_m\mapsto\te_m)=\bigcup_{i\in 1..m}\FV(\te_i)}
  \\ \\
  \ApplySubst{\x}{\subst}={\val}\ \mbox{if}\ {\subst(\x)=\val}\\
  \ApplySubst{\x}{\subst}=\x\ \mbox{if}\ \x\not\in\dom(\subst)\\
 \ApplySubst{\nval}{\subst}=\nval\\
  \ApplySubst{(\SumExp{\te_1}{\te_2})}{\subst}=\SumExp{\ApplySubst{\te_1}{\subst}}{\ApplySubst{\te_2}{\subst}}\\
  {\ApplySubst{({\LambdaExp{\x}{\te}})}{\subst}={\LambdaExp{\x}{\ApplySubst{\te}{\CancelSubst{\subst}{\{\x\}}}}}\ \mbox{if}\ \x\not\in\FV(\subst)}\\
  {\ApplySubst{(\AppExp{\te_1}{\te_2})}{\subst}=\AppExp{\ApplySubst{\te_1}{\subst}}{\ApplySubst{\te_2}{\subst}}}\\
  \ApplySubst{\Unbind{\tenv}{\te}}{\subst}=\Unbind{\tenv}{\ApplySubst{\te}{\CancelSubst{\subst}{\dom(\tenv)}}}\ \mbox{if}\ {\dom(\tenv)\cap\FV(\subst)=\emptyset}\\
  {\ApplySubst{\Rebind{\te}{\x_1{:}\T_1\mapsto\te_1,\ldots,\x_m{:}\T_m\mapsto\te_m}}{\subst}=\Rebind{\ApplySubst{\te}{\subst}}{\x_1{:}\T_1\mapsto\ApplySubst{\te_1}{\subst},\ldots,\x_m{:}\T_m\mapsto\ApplySubst{\te_m}{\subst}}}
\end{array}
$ \caption{Free variables and {application of substitution}}\label{fig:subst} \HSep
\end{figure}

The call-by-value operational semantics is described by the
reduction rules and the definition of the evaluation contexts
$\context$. We denote by $\tilde{\nval}$ the integer represented by
the constant $\nval$, {by $\extractTEnv{\tsubst}$ and
$\extractSubst{\tsubst}$ the type context and the untyped substitution\EZ{, respectively,}
extracted from a typed substitution $\tsubst$, by $\dom$ the domain
of a map,} by $\RestrictSubst{\subst}{\{\x_1,\ldots,\x_n\}}$ and
$\CancelSubst{\subst}{\{\x_1,\ldots,\x_n\}}$ the substitutions
obtained from $\subst$ by restricting to or removing variables in
set $\{\x_1,\ldots,\x_n\}$, respectively.

Rules for sum and application (of a lambda to a value) are standard.
The $\lab{Rebind\_}$ rules determine what happens when a rebind is
applied to a term. There are two rules for the rebinding of an
unbound term. Rule $\lab{{RebindUnbindYes}}$ is applied when the
unbound variables are all present (and of the required types), in
which case the associated values are substituted, otherwise rule
$\lab{{RebindUnbindNo}}$ produces a dynamic error. This is formally
expressed by the side condition
$\tenv\subseteq\extractTEnv{\tsubst}$. Note that a rebind
applied to a term may be stuck {even though} the variables  are all
present and of the right type, when the substitution is not defined.
This may be caused by the fact that when applied to an unbound term,
a substitution could cause capture of free variables (see the example
earlier in this section). \EZ{Rebinding a term which is a number or $\error$ does not affect the term.} On sum, {abstraction and application}, the
rebind is simply propagated to subterms, and if a rebind is applied
to a {rebound} term, $\lab{RebindRebind}$, the inner rebind is
applied first. The evaluation order is specified by rule $\lab{Ctx}$
and the definition of contexts, $\context$, that gives the
call-by-value strategy. Finally rule $\lab{CtxError}$ propagates
errors. To make rule selection deterministic, rules $\lab{Ctx}$ and
$\lab{CtxError}$ are applicable only when
$\context\not=\emptycontext$. As usual $\evstar$ is the reflexive
and transitive closure of $\ev$.

When a {rebind} is applied, only variables which were explicitly
specified as unbinders are replaced. For instance, the term
$\Rebind{\Unbind{\x}{\SumExp{\x}{\y}}}{\x\mapsto 1,\y\mapsto 2}$
reduces to $\SumExp{1}{\y}$ rather than to $\SumExp{1}{2}$. In other
words, the {unbinding/rebinding} mechanism is explicitly controlled by
the programmer.

\EZComm{cambiato come in ITA}Looking at the rules we can see that
\EZ{there is no rule for the rebinding of a variable}. Indeed, it will be resolved
only when the variable \EZ{is} substituted as effect of a standard
application. \EZ{For instance, the term  $\AppExp{(\LambdaExp{\y}{\Rebind{\y}{\x\mapsto 2})}}{\Unbind{\x}{\x+1}}$ reduces to
$\Rebind{\Unbind{\x}{\x+1}}{\x\mapsto 2}$.}

%Looking at the rules we can see that rebind remains stuck on a
%variable. Indeed, it will be resolved only when the variable will be
%substituted as effect of a standard application. See the following
%example:
%\begin{center}
%$\begin{array}{ll}
%\AppExp{\Rebind{(\LambdaExp{\y}{\y+\Unbind{\x}{\x})}}{\x\mapsto
%1}}{\Unbind{\x}{\x+2}}& \ev
%\AppExp{(\LambdaExp{\y}{\Rebind{(\y+\Unbind{\x}{\x})}{\x\mapsto 1})}}{\Unbind{\x}{\x+2}}\\
%&\ev
%\Rebind{(\Unbind{\x}{\x+2}+\Unbind{\x}{\x})}{\x\mapsto 1}\\
%&\ev\Rebind{\Unbind{\x}{\x+2}}{\x\mapsto
%1}+\Rebind{\Unbind{\x}{\x}}{\x\mapsto 1} \\
%&\evstar 4
%\end{array}
%$
%\end{center}
Note that in rule $\lab{RebindAbs}$, the binder $\x$ of the
$\lambda$-abstraction does not interfere with the rebind, even in
 case $\x\in\dom(\tsubst)$. Indeed, rebind has no effect on the
free occurrences of $\x$ in the body of the $\lambda$-abstraction.
For instance,
$\AppExp{\Rebind{(\LambdaExp{\x}{\x+\Unbind{\x}{\x}})}{\x\mapsto
1}}{2}$, which is  $\alpha$-equivalent to
$\AppExp{\Rebind{(\LambdaExp{\y}{\y+\Unbind{\x}{\x}})}{\x\mapsto
1}}{2}$, reduces in \PG{some} steps to $2+1$. On the other hand,
both $\lambda$-{binders} and unbinders prevent a substitution for
the corresponding variable from being propagated in their scope, for
instance:
%\vspace{1.5mm}
\begin{center}
$\begin{array}{l}
\Rebind{\Unbind{\x,\y}{\x+\LambdaExp{\x}{(\x+\y)}+\Unbind{\x}{\x+\y}}}{\x\mapsto
2, \y\mapsto 3}\ev 2+(\LambdaExp{\x}{\x+3})+\Unbind{\x}{\x+3}
\end{array}$
%\vspace{1.5mm}
\end{center}

A standard
(static) binder can also affect code to be dynamically rebound, when
it binds free variables in a substitution $\tsubst$, as shown by the
following example:
\begin{center}
$\begin{array}{l}
\AppExp{\AppExp{(\LambdaExp{\x}{\LambdaExp{\y}{\SumExp{\Rebind{\y}{\x\mapsto
\x}}{\x}}})}{1}}{\Unbind{\x}{\SumExp{\x}{2}}}\ev
\AppExp{(\LambdaExp{\y}{\SumExp{\Rebind{\y}{\x\mapsto
1}}{1}})}{\Unbind{\x}{\SumExp{\x}{2}}}\\
\ev \SumExp{\Rebind{\Unbind{\x}{\SumExp{\x}{2}}}{\x\mapsto 1}}1\ev \SumExp{\SumExp 1 2} 1.
\end{array}
$ \end{center} {Note
that in $[\x\mapsto\x]$ the two occurrences of $\x$ refer to
different variables. Indeed, the second is bound by the external
lambda whereas the first {one is a rebinder.}}

%% file: TypeSystem.tex
\section{Type system}\label{sect:typesystem}
We have three kinds of types: {\em primitive types} $\tau$, {\em
value types} $\V$, and {\em term types} $\T$; see
\refToFigure{types}.

\begin{figure}[h]
\HSep
\begin{grammatica}
\produzione{\T}{\TypeWithLevel{\tau}{\level}\mid\Inter{\T_1}{\T_2}{\Space\Space (\level\in\Nat)}}{{term type}}\\
\produzione{\V}{\TypeWithLevel{\tau}{0}\mid\Inter{\V}{\T}}{{value type}}\\
\produzione{\tau}{\intType\mid\codeType\mid\funType{\T}{\T'}}{{primitive type}}\\
\end{grammatica}
\caption{Types}\label{fig:types} \HSep
\end{figure}

Primitive types characterise the shape of values. In our case we have
integers ($\intType$), functions (\EZ{$\funType{\T}{\T'}$}), and
$\codeType$, which is the type of a term $\Unbind{\tenv}{\te}$, that
is, (possibly) open code.

Term types are primitive types decorated with a {\em level} $\level$
or intersection of types.  If a term has type
$\TypeWithLevel{\tau}{\level}$, then by applying $\level$ rebind
operators to the term we get a value of primitive type $\tau$. We
abbreviate a type $\TypeWithLevel{\tau}{0}$ by $\tau$. Terms have
the intersection type $\Inter{\T_1}{\T_2}$ when they have both types
$\T_1$ and $\T_2$. On intersection we have the usual congruence due
to idempotence, commutativity, associativity, and distributivity
over arrow type, defined in the first four clauses of
\refToFigure{congruence}.

Value types characterise terms that reduce to values, so they are
intersections in which (at least) one of the conjuncts must be a
primitive type {of level 0.} For instance, the term
$\Unbind{\x:\intType}{\Unbind{\y:\intType}{\SumExp{\x}{\y}}}$ has
type
$\Inter{\Inter{\TypeWithLevel{\codeType}{0}}{\TypeWithLevel{\codeType}{1}}}{\TypeWithLevel{\intType}{2}}$,
since it is code that applying one rebinding produces code that, in
turn, applying another rebinding produces an integer. The term
$\Unbind{\x:\intType}{\SumExp{\x}{\Unbind{\y:\intType}{\SumExp{\y}{1}}}}$
has type
$\Inter{\TypeWithLevel{\codeType}{0}}{\TypeWithLevel{\intType}{2}}$
since it is code that applying one rebinding produces the term
${\SumExp{\nval}{\Unbind{\y:\intType}{\SumExp{\y}{1}}}}$, for some
$\nval$. Both
$\Inter{\Inter{\TypeWithLevel{\codeType}{0}}{\TypeWithLevel{\codeType}{1}}}{\TypeWithLevel{\intType}{2}}$
and
$\Inter{\TypeWithLevel{\codeType}{0}}{\TypeWithLevel{\intType}{2}}$
are value types, whereas ${\TypeWithLevel{\intType}{{1}}}$,
which is the type of term
${\SumExp{\nval}{\Unbind{\y:\intType}{\SumExp{\y}{1}}}}$, is not a
value type. Indeed, in order to produce an integer value the term
must be rebound (at least) once. The typing rule for application
enforces the restriction that a term may be applied only to terms
reducing to values, that is the call-by-value strategy. Similar for
the terms associated with variables in a substitution.

Let $I=\{1,\ldots,m\}$. We write $\bigwedge_{i\in
I}\TypeWithLevel{\tau_i}{\level_i}$ and $\bigwedge_{i\in 1..m}\TypeWithLevel{\tau_i}{\level_i}$ to denote
$\TypeWithLevel{\tau_1}{\level_1}\wedge\cdots\wedge\TypeWithLevel{\tau_m}{\level_m}$.
Note that any type $\T$ is such that $\T=\bigwedge_{i\in 1..m}\TypeWithLevel{\tau_i}{\level_i}$, for some $\tau_i$ and
$\level_i$ ($i\in 1..m$). Given a type $\T=\bigwedge_{i\in 1..m}\TypeWithLevel{\tau_i}{\level_i}$, with
$\incrLevel{\T}{\levelPrime}$ we denote the type $\bigwedge_{i\in 1..m}\TypeWithLevel{\tau_i}{\level_i+\levelPrime}$.
\begin{figure}[h] \hrule
\[
\begin{array}{c}

\begin{array}{c@{\quad\quad\quad\quad}c@{\quad\quad\quad\quad}c}
  \T\congr\Inter{\T}{\T} & \Inter{\T_1}{\T_2}\congr\Inter{\T_2}{\T_1} &
  \Inter{\T_1}{(\Inter{\T_2}{\T_3})}\congr\Inter{(\Inter{\T_1}{\T_2})}{\T_3}
\end{array} \\ \\

\Inter {\TypeWithLevel{(\funType{\T}{\T_1})}{\level}}
{\TypeWithLevel{(\funType{\T}{\T_2})}{\level}} \congr
\TypeWithLevel{(\funType{\T}{\Inter{\T_1}{\T_2}})}{\level} \qquad \TypeWithLevel{(\funType{\T'}{\incrLevel{\T}{\levelPrime}})}{\level+1}\congr\TypeWithLevel{(\funType{\T'}{\incrLevel{\T}{(\levelPrime+1)}})}{\level}\\ \\
\end{array}
\]
 \hrule
 \caption{Congruence on types}\label{fig:congruence}
 \end{figure}

\refToFigure{congruence} defines congruence on types.
In addition to the standard properties of intersection, the last
\EZ{clause} says that the level of function types can be switched
with the one of their results. That is, unbinding and
lambda-abstraction commute. So rebinding may be applied to
lambda-abstractions, since reduction rule $\lab{RebindAbs}$
pushes rebinding inside. For instance, the terms
$\Unbind{\tenv}{\LambdaExp{\x}{\te}}$ and
$\LambdaExp{\x}{\Unbind{\tenv}{{\te}}}$ may be used interchangeably.
  \begin{figure}[h] \hrule
\[
\begin{array}{c}
\begin{array}{c@{\quad\quad\quad\quad}c@{\quad\quad\quad\quad}c}
  \intType^\level\leq\intType^{\level+1} &
%  \TypeWithLevel{(\funType{\V}{\TypeWithLevel{\tau}{\levelPrime+1}})}{\level}\leq\TypeWithLevel{(\funType{\V}{\TypeWithLevel{\tau}{\levelPrime}})}{\level+1}
%  &
   \Inter{\T_1}{\T_2}\leq{\T_1}
%\end{array}
\\ \\
%\begin{array}{c@{\quad\quad\quad\quad}c}
 \prooftree
    \T_2\leq\T_1\quad\quad\T'_1\leq\T'_2
    \justifies
  {\TypeWithLevel{(\funType{\T_1}{\T'_1})}{\level}}\leq
  {\TypeWithLevel{(\funType{\T_2}{\T'_2})}{\level}}
  \endprooftree
  &
 \prooftree
    \T_1\leq\T_1'\quad\quad\T_2\leq\T_2'
    \justifies
   \Inter{\T_1}{\T_2}\leq \Inter{\T'_1}{\T'_2}
  \endprooftree
% \end{array}
 \\ \\
% \begin{array}{c@{\quad\quad\quad\quad}c}
%\T\leq\T &
 \prooftree
    \T_1\leq\T_2\quad\quad\T_2\leq\T_3
    \justifies
   \T_1\leq \T_3
         \endprooftree &
    \prooftree
    \T_1\congr\T_2
    \justifies
   \T_1\leq \T_2

  \endprooftree
\end{array}
\end{array}
\]
 \hrule
 \caption{Subtyping on types}\label{fig:subtyping}
 \end{figure}

Subtyping, defined in \refToFigure{subtyping}, expresses subsumption,
that is, if a term has type $\T_1$, then it can be used also in a
context requiring a type $\T_2$ with $\T_1\leq\T_2$. For
integer types it is justified by the reduction rule
$\lab{RebindNum}$, since once we obtain {an integer value} any number of
rebindings may be applied.\footnote{Note that the generalisation of $\intType^\level\leq\intType^{\level+1}$ to  $\T^\level\leq\T^{\level+1}$ is sound but useless.}
For intersections, it is
intersection elimination. The other rules are the standard extension
of subtyping to function and intersection types,
transitivity, and the fact that congruent types are in the subtyping
relation.

  \begin{figure}[h]
\HSep
  \begin{center}
$\begin{array}{l}
 \NamedRule{T-Inter}{\IsWFExp{\Gamma}{\te}{\T_1}\Space\IsWFExp{\Gamma}{\te}{\T_2}}{\IsWFExp{\Gamma}{\te}{\Inter{\T_1}{\T_2}}}{}\Space
 \NamedRule{T-Sub}{\IsWFExp{\Gamma}{\te}{\T}\Space \T\leq\T'}{\IsWFExp{\Gamma}{\te}{\T'}}{}\Space
\NamedRule{T-Var}{\Gamma(\x)=\T}{\IsWFExp{\Gamma}{\x}{\T}}{}\\[4ex]
\NamedRule{T-Num}{}{\IsWFExp{\Gamma}{\nval}{\TypeWithLevel{\intType}{0}}}{}\Space
\NamedRule{T-Sum}{\IsWFExp{\Gamma}{\te_1}{\TypeWithLevel{\intType}{\level}}\Space\IsWFExp{\Gamma}{\te_2}{\TypeWithLevel{\intType}{\level}}}{\IsWFExp{\Gamma}{\SumExp{\te_1}{\te_2}}{\TypeWithLevel{\intType}{\level}}}{}\Space\NamedRule{T-Error}{}{\IsWFExp{\Gamma}{\error}{\T}}{}
\\[4ex]
\NamedRule{T-Abs}{\IsWFExp{\Subst{\tenv}{\x{:}\T}}{\te}{\T'}}{\IsWFExp{\tenv}{\LambdaExp{\x}{\te}}{\TypeWithLevel{(\funType{\T}{\T'})}{0}}}{}\Space
\NamedRule{T-App}
{\IsWFExp{\Gamma}{\te_1}{\TypeWithLevel{(\funType{{\V}}{\T})}{0}}\Space
\IsWFExp{\Gamma}{\te_2}{{\V}}} {\IsWFExp{\Gamma}{\AppExp{\te_1}{\te_2}}{\T}}{}\\[4ex]
\NamedRule{T-Unbind-0}{\IsWFExp{\Subst{\tenv}{\tenv'}}{\te}{\T}}{\IsWFExp{\tenv}{\Unbind{\tenv'}{\te}}{\TypeWithLevel{\codeType}{0}}}{}\Space
 \NamedRule{T-Unbind}{\IsWFExp{\Subst{\tenv}{\tenv'}}{\te}{\T}}{\IsWFExp{\tenv}{\Unbind{\tenv'}{\te}}{\incrLevel{\T}{1}}}{}\\[4ex]
\NamedRule{T-Rebind} {\IsWFExp{\tenv}{\te}{\incrLevel{\T}{1}}\Space
\IsWFSubst{\tenv}{\tsubst}{\OK}}
{\IsWFExp{\tenv}{\Rebind{\te}{\tsubst}}{\T}}{}\Space
\NamedRule{T-Rebinding}{\IsWFExp{\tenv}{\te_i}{\V_i}\Space \V_i\leq\T_i\Space(i\in
1..m)}{\IsWFSubst{\tenv}{\x_1{:}\T_1\mapsto\te_1,\ldots,\x_m{:}\T_m\mapsto\te_m}{\OK}}{}
\end{array}$
\end{center}
\caption{Typing rules}\label{fig:typing} \HSep
\end{figure}

Typing rules are defined in \refToFigure{typing}.
We denote by $\Subst{\tenv}{\tenv'}$ the concatenation of the two type contexts $\tenv$ and $\tenv'$ with disjoint domains, which turns out to be a type context (map) as well.

A number has the
value type $\TypeWithLevel{\intType}{{0}}$. With rule $\lab{T-Sub}$, however, it can be
given the type $\TypeWithLevel{\intType}{\level}$ for any $\level$.
Rule $\lab{T-Sum}$ requires that both operands of a sum have the same
type,  with rule $\lab{T-Sub}$ the term can be given as level the
biggest level of the operands. Rule $\lab{T-Error}$ permits the use of
$\error$ in any context. In rule $\lab{T-Abs}$ the initial level of a
lambda abstraction is $0$ since the term is a value. With rule
$\lab{T-Sub}$ we may decrease the level of the return type by
increasing, by the same amount, the level of the whole arrow type.
This
%is, on one side, in accord with rule $\lab{T-App}$ where the
%level of the type of an application is the sum of these two levels. On the other, it
is useful since, for example, we can derive
\[\IsWFExp{}{\LambdaExp{\x}
{\SumExp\x{\Unbind{\y{:}\intType}{\SumExp\y{\Unbind{\z{:}\intType}{\z}}}}}}
{\TypeWithLevel{(\funType{\intType}{\TypeWithLevel{\intType}{1}})}{1}}
\]
by first deriving the type
${\TypeWithLevel{(\funType{\intType}{\TypeWithLevel{\intType}{2}})}{0}}$
for the term, and then applying $\lab{T-Sub}$. Therefore, we can give
type to the rebinding of the term, by applying rule $\lab{T-Rebinding}$
that requires that the term to be rebound has level bigger than
$0$, and whose resulting type is decreased by one. For example,
\[
\IsWFExp{}{\Rebind{(\LambdaExp{\x}
{\SumExp\x{\Unbind{\y{:}\intType}{\SumExp\y{\Unbind{\z{:}\intType}{\z}}}}})}
{\y{:}\intType\mapsto
5}}{\TypeWithLevel{(\funType{\intType}{\TypeWithLevel{\intType}{1}})}{0}}
\]
which means that the term reduces to  a lambda abstraction, i.e., to
a value, which applied to an integer needs one {rebind} in order to
produce an integer or error. The rule $\lab{T-App}$ assumes that the
type of the function be a  level $0$ type. This is not a
restriction, since using rule $\lab{T-Sub}$, if the term has any
function type it is possible to assign it a level $0$ type. {The
type of the argument must be a value type.  This condition is
justified by the example given in the introduction.}

The two rules for unbinds reflect the fact that code is both a
value, and as such has \MD{type $\codeType^0$,} %{$\codeType$ of} level $0$ type,
and also a
term that needs one more rebinding than its body in order to produce
a value. {Taking the intersection of the types derived for the
same unbind with these two rules we can derive a value type for the
unbind and use it as argument of an application. For example typing
$\Unbind{ y:{\tt int}}{y}$ by $\tt code^0\wedge int^1$ we can derive
type $\TypeWithLevel{\intType}{{0}}$ for the term
\begin{center}{($\lambda x. 2 + x [ y$: {\tt int} $\mapsto 3])\Unbind{
y:{\tt int}}{y}$.}
\end{center}}

The present type system only takes into account the number of
rebindings which are applied to a term, whereas  no check is
performed on the name and the type of the variables to be rebound.
This check is performed at runtime by rules $\lab{RebindUnbindYes}$
and $\lab{RebindUnbindNo}$. \EZ{In {papers introducing the calculus} \cite{DezaniEtAl09,DezaniEtAl10} we have
also provided an alternative type system (for the call-by-name
calculus) such that this check is statically perfomed, see the Conclusion for more comments.}

Note that terms which are stuck since application of substitution is
undefined, such as  the previous example
$\AppExp{(\LambdaExp{\y}{\Unbind{\x}{\y}})}{(\LambdaExp{\z}{\x})}$,
are ill typed. Indeed, in the typing rules for unbinding, unbinders  are required to be disjoint from outer binders, and there is no weakening rule. Hence, a type context for the example should simultaneously include a type for $\x$ and not include a type for $\x$ in order to type $\LambdaExp{\z}{\x}$ and $\LambdaExp{\y}{\Unbind{\x}{\y}}$, respectively. For the same reason,
a peculiarity of the given type system is that weakening does not hold, in spite of the fact that no notion of linearity is enforced.

%% file: Soundness.tex
\section{Soundness of the type system}\label{sect:soundness}
%\section{{Soundness proofs}}\label{sect:soundness}

The type system is {\em safe} since types are preserved by reduction
and a closed term with value type is a value or \EZ{$\error$ or} can be reduced. In
other words the system has both the {\em subject reduction} and the
{\em progress} properties. Note that a term that may not be
assigned a value type \MD{can be} stuck, as for example
$\SumExp1{\Unbind{\x{:}\intType}{\x}}$, which has type
$\TypeWithLevel{\intType}{1}$.

The proof of subject reduction (Theorem \ref{srt}) is standard. We
start with a lemma (Lemma \ref{l:subtypingArrow}) on the properties
of the equivalence and pre-order relation on types, which can be
easily shown by induction on their definitions. Then we give an
Inversion Lemma (Lemma \ref{il}), a Substitution Lemma (Lemma
\ref{sl}) and a Context Lemma (Lemma \ref{cont}). Lemmas \ref{il}
and \ref{sl} can be easily shown by induction on {type derivations}.
The proof of Lemma \ref{cont} is by structural induction on
contexts.

\begin{lemma}\label{l:subtypingArrow}\
\begin{enumerate}
\item\label{l:subtypingArrow1} If $\bigwedge_{
i\in I}\TypeWithLevel{(\funType{\T_i}{\T'_i})}{0}\equiv\TypeWithLevel{(\funType{\T}{\T'})}{0}$, then
$\T_i\equiv\T$ for all $i\in I$ and $\bigwedge_{
i\in I}\T'_i\equiv\T'$.
\item\label{l:subtypingArrow2} If $\bigwedge_{
i\in I}\TypeWithLevel{(\funType{\T_i}{\T'_i})}{0}\leq\T$, then there
are $L$, $\hat\T_l$, $\hat\T'_l$, $(l\in L)$, such that
$\T\congr\bigwedge_{l\in
L}\TypeWithLevel{(\funType{\hat\T_l}{\hat\T'_l})}{0}, $ and for all
$l\in L$ there is $J_l\subseteq \EZ{I}$\EZComm{giusto?} with
\begin{itemize}
  \item $\hat\T_l\leq\T_j$ for all $j\in J_l$, and
  \item $\bigwedge_{j\in J_l}\hat\T'_j\leq \T'_l$.\MDComm{ho aggiunto dei cappelli}
\end{itemize}
\end{enumerate}
\end{lemma}

\begin{lemma}[Inversion Lemma]\label{il}\
\begin{enumerate}
\item\label{il1} If $\IsWFExp{\Gamma}{\x}{\T}$, then $\Gamma(\x)\leq\T$.
\item\label{il2} If $\IsWFExp{\Gamma}{\nval}{\T}$, then {\em $\TypeWithLevel{\intType}{0}\leq\T$.}
\item\label{il3} If {\em $\IsWFExp{\Gamma}{\SumExp{\te_1}{\te_2}}{\T}$,} then {\em $\TypeWithLevel{\intType}{0}\leq\T$} and {\em $\IsWFExp{\Gamma}{\te_1}{\T}$} and {\em $\IsWFExp{\Gamma}{\te_2}{\T}$.}
\item\label{il4} If $\IsWFExp{\Gamma}{\LambdaExp{\x}{\te}}{\T}$,
then there are $m$, $\T_i$, $\T'_i$ ($i\in 1..m$) such that
$\T\equiv\bigwedge_{i\in 1..m}\TypeWithLevel{(\funType{\T_i}{\T'_i})}{0}$, and $\IsWFExp{\Subst{\tenv}{\x{:}\T_i}}{\te}{\T'_i}$ ($i\in 1..m$).
\item\label{il5} If $\IsWFExp{\Gamma}{\AppExp{\te_1}{\te_2}}{\T}$, then
there is $\V$ such that
$\IsWFExp{\Gamma}{\te_1}{\TypeWithLevel{(\funType{\V}{\T})}{0}}$ and
$\IsWFExp{\Gamma}{\te_2}{\V}$.
\item\label{il6} If $\IsWFExp{\Gamma}{\Unbind{\tenv'}{\te}}{\bigwedge_{i\in 1..m}\TypeWithLevel{\tau_i}{\level_i}}$, then
\begin{itemize}
\item $\IsWFExp{\Subst{\tenv}{\tenv'}}{\te}{\TypeWithLevel{\tau_i}{\level_i-1}}$ for all
$\level_i>0$, and
\item $\tau_i=\codeType $ for all $\level_i=0$.
\end{itemize}
%In the particular case $m=1$ and $\level_1=0$ we have
%$\IsWFExp{\Subst{\tenv}{\tenv'}}{\te}{\T'}$, for some
%$\T'$.
When $m=1$ and $\level_1=0$ we also have
$\IsWFExp{\Subst{\tenv}{\tenv'}}{\te}{\T'}$, for some $\T'$.
\item\label{il7} If $\IsWFExp{\Gamma}{\Rebind{\te}{\tsubst}}{\bigwedge_{i\in 1..m}\TypeWithLevel{\tau_i}{\level_i}}$, then $\IsWFExp{\tenv}{\te}{\bigwedge_{i\in 1..m}\TypeWithLevel{\tau_i}{\level_i+1}}$ and $\IsWFSubst{\tenv}{\tsubst}{\OK}$.
\item\label{il8} If $\IsWFSubst{\tenv}{\x_1{:}\T_1\mapsto\te_1,\ldots,\x_m{:}\T_m\mapsto\te_m}{\OK}$, then there are $\V_i\leq\T_i$ for $i\in 1..m$ such that $\IsWFExp{\tenv}{\te_i}{\V_i}$.
\end{enumerate}
\end{lemma}
\begin{proof}
By induction on typing derivations. We only consider some
interesting cases.

For Point (\ref{il4}), if the last applied rule is \lab{T-Sub}, then  the
result follows by induction from Lemma
\ref{l:subtypingArrow}(\ref{l:subtypingArrow2}). For the same Point
if the last applied rule is \lab{T-Inter}, then let
$\IsWFExp{\Gamma}{\LambdaExp{\x}{\te}}{\Inter{\T}{T'}}$. By
induction hypothesis there are $m$, $m'$, $\T_i$, $\T'_j$,
$\hat{\T}_i$, $\hat{\T}'_j$ ($i\in 1..m$, $j\in 1..m'$) such that:
\begin{itemize}
  \item $\T\equiv\bigwedge_{i\in 1..m}\TypeWithLevel{(\funType{\T_i}{\hat{\T}_i})}{0}$,
  \item $\IsWFExp{\Subst{\tenv}{\x{:}\T_i}}{\te}{\hat{\T}_i}$ ($i\in 1..m$),
  \item $\T'\equiv\bigwedge_{j\in 1..m'}\TypeWithLevel{(\funType{\T'_j}{\hat{\T}'_j})}{0}$, and
  \item $\IsWFExp{\Subst{\tenv}{\x{:}\T'_j}}{\te}{\hat{\T}'_j}$ ($j\in 1..m'$).
\end{itemize}
Therefore \EZ{$\Inter{\T}{\T'}\equiv\bigwedge_{i\in
1..m}\TypeWithLevel{(\funType{\T_i}{\hat{\T}_i})}{0}\wedge\bigwedge_{j\in
1..m'}\TypeWithLevel{(\funType{\T'_j}{\hat{\T}'_j})}{0}$.}\EZComm{corretto il penultimo pedice da $i$ a $j$}

For Point (\ref{il5}) if the last applied rule is \lab{T-Inter},
\PG{then let
$\IsWFExp{\Gamma}{\AppExp{\te_1}{\te_2}}{\Inter{\T_1}{T_2}}$. By
induction hypothesis we have
$\IsWFExp{\Gamma}{\te_1}{\funType{\V_i}{\EZ{\T_i}}}$ and
$\IsWFExp{\Gamma}{\te_2}{\V_i}$, for some $\V_i$,
$i=1,2$.}\PGComm{$T_1$ e $T_2$ non sono arbitrari} We derive
$\IsWFExp{\Gamma}{\te_1}{(\funType{\V_1}{\T_1})\wedge(\funType{\V_2}{\T_2})}$
and $\IsWFExp{\tenv}{\te_2}{\V_1\wedge\V_2}$ by rule \lab{T-Inter}.
\PG{Since}
$(\funType{\V_1}{\T_1})\wedge(\funType{\V_2}{\T_2})\leq\funType{\V_1\wedge\V_2}{\T_1\wedge\T_2}$,
\PG{applying rule \lab{T-Sub} we get}
$\IsWFExp{\Gamma}{\te_1}{\funType{\V_1\wedge\V_2}{\T_1\wedge\T_2}}$.\MDComm{ho
eliminato tutti i $'$}

\end{proof}

\begin{lemma}[Substitution Lemma]\label{sl}
If $\IsWFExp{\Subst{\tenv}{\x{:}\T}}{\te}{\T'}$,  and
$\IsWFExp{\tenv}{\val}{\T}$, then
$\IsWFExp{\tenv}{\ApplySubst{\te}{\x\mapsto\val}}{\T'}$.
\end{lemma}

\begin{lemma}[Context Lemma]\label{cont}
Let $\IsWFExp{\tenv}{\inContext{\te}}{\T}$, then
\begin{itemize}
\item $\IsWFExp{\tenv}{\te}{\T'}$ for some $\T'$, and
\item if $\IsWFExp{\tenv}{\te'}{\T'}$, then
$\IsWFExp{\tenv}{\inContext{\te'}}{\T}$, for all $\te'$.
\end{itemize}
\end{lemma}

\begin{theorem}[Subject Reduction]\label{srt}
If $\IsWFExp{\tenv}{\te}{\T}$ and $\te\EZ{\ev}\te'$, then $\IsWFExp{\tenv}{\te'}{\T}$.\MDComm{non sono d'accordo su eliminare lo star, la SR vale per lo star!}
\end{theorem}
\begin{proof}
By induction on reduction derivations. We
only consider some interesting cases.\\

If the applied rule is $\lab{App}$, then
\[{\AppExp{(\LambdaExp{\x}{\te})}{\val}\ev\ApplySubst{\te}{\x\mapsto\val}}\]
From
$\IsWFExp{\tenv}{\AppExp{(\LambdaExp{\x}{\te})}{\val}}{\T}$
by Lemma \ref{il}, case (\ref{il5}) we have that: there is $\V$ such
that
$\IsWFExp{\Gamma}{\LambdaExp{\x}{\te}}{\TypeWithLevel{(\funType{\V}{\T})}{0}}$
and $\IsWFExp{\Gamma}{\val}{\V}$. By Lemma \ref{il}, case
(\ref{il4}) we have that there are $m$, $\T_i$, $\T'_i$ ($i\in 1..m$) such that
  $\TypeWithLevel{(\funType{\V}{\T})}{0}\equiv\bigwedge_{i\in 1..m}\TypeWithLevel{(\funType{\T_i}{\T'_i})}{0}$, and $\IsWFExp{\Subst{\tenv}{\x{:}\T_i}}{\te}{\T'_i}$ ($i\in 1..m$).
From Lemma \ref{l:subtypingArrow}(\ref{l:subtypingArrow1}) we get $\T_i\equiv\V$ for all $i\in 1..m$ and $\bigwedge_{
i\in 1..m}\T'_i\equiv\T$.
Then we can derive
$\IsWFExp{\Subst{\tenv}{\x{:}\V}}{\te}{\T}$ using rules \lab{T-Sub} and \lab{T-Inter}. By Lemma  \ref{sl} we conclude
that $\IsWFExp{\tenv}{\ApplySubst{\te}{x\mapsto \val}}{\T}$.\\

If the applied rule is $\lab{RebindUnbindYes}$, then
\[{\Rebind{\Unbind{\tenv'}{\te}}{\tsubstval}\ev\ApplySubst{\te}{{\RestrictSubst{\extractSubst{\tsubstval}}{\dom(\EZ{\tenv'})}}}}\qquad \tenv'\subseteq\extractTEnv{\tsubstval}\]
\EZ{Let}
$\RestrictSubst{\tsubstval}{{\dom(\tenv')}}=\x_1{:}\T_1\mapsto\val_1,\ldots,\x_m{:}\T_m\mapsto\val_m$.
Since $\tenv'\subseteq\extractTEnv{\tsubstval}$ we have that
$\tenv'=\EZ{\x_1{:}\T_1,\ldots,\x_m{:}\T_m}$. From
$\IsWFExp{\tenv}{\Rebind{\Unbind{\tenv'}{\te}}{\tsubstval}}{\T}$
by Lemma \ref{il}, case (\ref{il7}), we get
$\T=\bigwedge_{i\in 1..n}\TypeWithLevel{\tau_i}{\level_i}$ and
$\IsWFExp{\tenv}{\Unbind{\tenv'}{\te}}{\bigwedge_{i\in 1..n}\TypeWithLevel{\tau_i}{\level_i+1}}$
and $\IsWFSubst{\tenv}{{\tsubstval}}{\OK}$. From Lemma \ref{il}, case
(\ref{il6}),  we have that
$\IsWFExp{\Subst{\tenv}{\tenv'}}{\te}{\bigwedge_{i\in 1..n}\TypeWithLevel{\tau_i}{\level_i}}$.
Moreover, by Lemma \ref{il}, case (\ref{il8}), and rule $\lab{T-Sub}$ we have that
$\IsWFSubst{\tenv}{{\tsubstval}}{\OK}$ implies that
$\IsWFExp{\tenv}{\val_i}{\T_i}$ for $i\in 1..m$. Applying $m$
times Lemma \ref{sl}, we derive
$\IsWFExp{\tenv}{\ApplySubst{\te}{{\RestrictSubst{\extractSubst{\tsubstval}}{\dom(\EZ{\tenv'})}}}}{\T}$.
\end{proof}

In order to show the Progress Theorem (Theorem \ref{pt}), we start
as usual with a Canonical {Forms} Lemma (Lemma \ref{cfl}). Then we
state the standard relation between type contexts and free variables
(Lemma \ref{tcfv}), and lastly we prove that all closed terms which
are {rebound terms} always reduce (Lemma \ref{l:rebindMT}).

\begin{lemma}[Canonical Forms]\label{cfl}\

\begin{enumerate}
\item\label{cfl1} If $\IsWFExp{}{\val}{\TypeWithLevel{\text{\tt int}}{0}}$, then
$\val=\nval$.
\item\label{cfl2} If $\IsWFExp{}{\val}{\TypeWithLevel{\text{\tt code}}{0}}$, then $\val=\Unbind{\tenv}{\te}$.
\item\label{cfl3} If $\IsWFExp{}{\val}{\TypeWithLevel{(\funType{\T}{\T'})}{0}}$, then
$\val=\LambdaExp{\x}{\te}$.
\end{enumerate}
\end{lemma}
\begin{proof}
By case analysis on the \EZ{shape} of values.
\end{proof}

\begin{lemma}\label{tcfv}
If $\IsWFExp{\tenv}{\te}{\T}$, then
$\FV(\te)\subseteq\dom(\tenv)$.
\end{lemma}
\begin{proof}
By induction on type derivations.
\end{proof}

\begin{lemma}\label{l:rebindMT}
If $\te={\Rebind{\te'}{\tsubstval}}$ for some $\te'$ and $\tsubstval$, and
$\FV(\te)=\emptyset$, then $\te\ev\te''$ for some $\te''$.\end{lemma}
\begin{proof}
Let $\te=\Rebind{\te'}{\tsubstval_1}\cdots[\tsubstval_n]$
for some $\te'$, $\tsubstval_1$, \ldots, $\tsubstval_n$ ($n\geq 1$), where
$\te'$ is not a
rebind. The proof is by mathematical induction on $n$.\\
If $n=1$, then one of the reduction rules is applicable to
$\Rebind{\te'}{\tsubstval_1}$. Note that, if
$\te'=\Unbind{\tenv}{\te_1}$, then rule $\lab{{RebindUnbindYes}}$
is applicable in case $\tenv$ is {a} subset of the type environment
{associated} with
$\tsubstval_1$, otherwise rule $\lab{RebindUnbindNo}$ is applicable. \\
Let $\te=\Rebind{\te'}{\tsubstval_1}\cdots[\tsubstval_{n+1}]$. If
$\FV(\Rebind{\te'}{\tsubstval_1}\cdots[\tsubstval_{n+1}])=\emptyset$, then
also $\FV(\Rebind{\te'}{\tsubstval_1}\cdots[\tsubstval_{n}])=\emptyset$. By
induction hypothesis
$\Rebind{\te'}{\tsubstval_1}\cdots[\tsubstval_{n}]\ev\te''$, therefore
$\Rebind{\te'}{\tsubstval_1}\cdots[\tsubstval_{n+1}]\ev\Rebind{\te''}{\tsubstval_{n+1}}$
with rule $\lab{RebindRebind}$.
\end{proof}

\begin{theorem}[Progress]\label{pt}
If $\IsWFExp{}{\te}{\V}$, then either $\te$ is a value, or
$\te=\error$, or $\te\ev\te'$ for some $\te'$.
\end{theorem}
\begin{proof}
By induction on the derivation of $\IsWFExp{}{\te}{\V}$ with case analysis on the last typing
rule used.\\

If $\te$ is not a value or $\error$, {then} the last  applied rule in the
type derivation cannot be $\lab{T-Num}$, $\lab{T-Error}$,
$\lab{T-Abs}$, $\lab{T-Unbind-0}$, or $\lab{T-Unbind}$. Moreover  the typing environment for the
expression is empty, {hence} by Lemma \ref{tcfv}  the last applied rule  cannot be $\lab{T-Var}$.\\

If the last applied rule is \EZ{$\lab{T-Sub}$}, note that $\T\leq\V$
implies that $\T$ is a value type, and therefore the theorem holds
by induction.\\

If the last applied rule is $\lab{T-App}$, then
$\te={\AppExp{\te_1}{\te_2}}$, and taking into account that the resulting type must be a value type:
 \[\Rule
{\IsWFExp{}{\te_1}{\funType{\V'}{\V}}\Space \IsWFExp{}{\te_2}{\V'}}
{\IsWFExp{}{\AppExp{\te_1}{\te_2}}{\V}}{}\]
 If $\te_1$ is not a
value \EZ{or $\error$}, then, by induction hypothesis, $\te_1\ev\te_1'$. So
$\AppExp{\te_1}{\te_2}= \inContext{\te_1}$ with
$\context=\AppExp{\emptycontext}{\te_2}$, and by rule $\lab{Ctx}$,
$\AppExp{\te_1}{\te_2}\ev\AppExp{\te'_1}{\te_2}$. \EZ{If $\te_1$ is $\error$, we can apply rule $\lab{ContError}$ with the same context.} It $\te_1$ is a
value $\val$, \EZ{and} $\te_2$ is not a value \EZ{or $\error$}, then, by induction
hypothesis, $\te_2\ev\te_2'$. So $\AppExp{\te_1}{\te_2}=
\inContext{\te_2}$ with $\context=\AppExp\val{\emptycontext}$, and
by rule $\lab{Ctx}$, $\AppExp{\val}{\te_2}\ev\AppExp{\val}{\te'_2}$. \EZ{If $\te_2$ is $\error$, we can apply rule $\lab{ContError}$ with the same context.}
If both $\te_1$ and $\te_2$ are values, then by Lemma \ref{cfl},
case (\ref{cfl3}), $\te_1=\LambdaExp{\x}{\te'}$ and, therefore, we
can apply rule $\lab{App}$.
\\

If the last applied rule is $\lab{T-Sum}$, then
$\te=\SumExp{\te_1}{\te_2}$ and taking into account that the
resulting type must be a value type:
\[
\Rule{\IsWFExp{}{\te_1}{\TypeWithLevel{\intType}{0}}\Space\IsWFExp{}{\te_2}
{\TypeWithLevel{\intType}{0}}}{\IsWFExp{}{\SumExp{\te_1}{\te_2}}{\TypeWithLevel{\intType}{0}}}{}
\]
If $\te_1$ is not a value \EZ{or $\error$}, {then,} by induction hypothesis,
$\te_1\ev\te_1'$. So by rule \lab{Ctx}, with context
$\context=\SumExp{\emptycontext}{\te_2}$, we have
${\SumExp{\te_1}{\te_2}}\ev{\SumExp{\te'_1}{\te_2}}$. \EZ{If $\te_1$ is $\error$, we can apply rule $\lab{ContError}$ with the same context.} If $\te_1$
is a value, {then,} by Lemma \ref{cfl},  case (\ref{cfl1}), $\te_1={\nval_1}$. Now, if
$\te_2$ is not a value \EZ{or $\error$}, {then,} by induction hypothesis, $\te_2\ev\te_2'$.
So by rule \lab{Ctx}, with context
$\context=\SumExp{{\nval_1}}{\emptycontext}$, we get
${\SumExp{\te_1}{\te_2}}\ev{\SumExp{\te_1}{\te'_2}}$. \EZ{If $\te_2$ is $\error$, we can apply rule $\lab{ContError}$ with the same context.} Finally,
if $\te_2$ is a value, \EZ{then}  by Lemma \ref{cfl}, case (\ref{cfl1}),
$\te_2={\nval_2}$. Therefore rule \lab{Sum} is applicable.
\\

If the last applied rule is \lab{T-Rebind}, then
$\te={\Rebind{\te'}{\tsubst}}$. If some term $\te_i$ in $\tsubst$ is  not a value, then by Lemma \ref{il}(\ref{il7}) and (\ref{il8})
$\te_i$ is typed with a value type, and therefore $\te_i\ev\te_i'$ by induction, so $\te$ reduces using rule \lab{Ctx}.
Otherwise $\te=\Rebind{\te'}{\tsubstval}$. Since $\IsWFExp{}
{\Rebind{\te'}{\tsubstval}}{\V}{}$, we have
that $\FV(\Rebind{\te'}{\tsubstval})=\emptyset$ by Lemma \ref{tcfv}. From Lemma
\ref{l:rebindMT} we get that $\Rebind{\te'}{\tsubstval}\ev
\te''$ for some $\te''$.
\end{proof}

%% file: Conclu.tex
\section{Conclusion}\label{sect:rwc}

We have defined a type system with intersection types for an
extension of lambda-calculus with unbind and rebind operators
introduced in {previous work} \cite{DezaniEtAl09,DezaniEtAl10}. Besides the
traditional use of intersection types for {typing} (finitely)
polymorphic functions, this type system shows two novel
applications:
\begin{itemize}
\item An intersection type expresses that a term can be used in
contexts which provide a different number  of unbinds.
\item In particular, an unbound term can be used both as a value
of type $\codeType$ and in a context providing an unbind.
\end{itemize}

{This type system could be used for call-by-name with minor
modifications. However, the call-by-value case is more significant
since the condition that the argument of an application must reduce
to a value can be nicely expressed by the notion of value type.

Moreover, only the number of rebindings which are applied to a term
is taken into account, whereas  no check is performed on the name
and the type of the variables to be rebound; this check is performed
at runtime. This solution is convenient, e.g., in distributed
scenarios where code is not all available at compile time, or in
combination with delegation mechanisms where, in case of dynamic
error due to an absent/wrong binding, an alternative action is
taken. In {papers introducing the calculus} \cite{DezaniEtAl09,DezaniEtAl10} we have
also provided an alternative type system (for the call-by-name
calculus) which ensures a stronger form of safety, that is, that
rebinding always succeeds. The key idea is to decorate types
with the names of the variables which need to be rebound, as done
also by  Nanevski and Pfenning \cite{NanevskiPfenning05}. In this
way run-time errors arising from absence (or mismatch) in rebind are
prevented by a purely static type system, at the price of quite
sophisticated types. A similar system could be developed for the
present calculus. Also the type system of
Dezani et al. \cite{DezaniEtAl09,DezaniEtAl10} could be enriched
with intersection types to get the stronger safety for the
call-by-value calculus.}

Intersection types have been originally introduced \cite{coppdeza80} as a language for
descri\-bing and capturing properties of $\lambda$-terms, which had
escaped all previous typing disciplines. For instance, they were used
in order to give the first type theoretic characterisation of {\em
strongly normalising} terms \cite{pott80}, and later in order to capture {\em
(persistently) normalising terms} \cite{coppdezazacc87}.

Very early it was realised that intersection types had also a
distinctive semantical flavour. Namely, they expressed at a
syntactical level the fact that a term belonged to suitable compact
open sets in a Scott domain \cite{barecoppdeza83}.  Since then,
intersection types have been used as a powerful tool both for the
analysis and the synthesis of $\lambda$-models. On the one hand,
intersection type disciplines provide finitary inductive definitions
of interpretation of $\lambda$-terms in models
\cite{coppdezahonslong84}, and they are suggestive for the shape the
domain model has to have in order to exhibit specific
properties \cite{dezaghillika04}.

More recently, systems with both intersection and union types
have been proposed for various aims
\cite{Barbanera-Dezani-deLiguoro:IC-95,FCB08}, but we do not see any
gain in adding union types in the present setting.

Ever since the accidental discovery of dynamic scoping in McCarthy's
Lisp 1.0, there has been extensive work in explaining and
integrating {mechanisms for} dynamic and static binding. The
classical reference for dynamic scoping is  {Moreau's paper}
\cite{M98}, which introduces a $\lambda$-calculus with two distinct
kinds of variables: {\em static} and {\em dynamic}. The semantics
can be (equivalently) given either by translation in the standard
$\lambda$-calculus or directly. In the translation semantics,
$\lambda$-abstractions  have an additional parameter corresponding
to the application-time context. In the direct semantics, roughly,
an application $\AppExp{(\LambdaExp{\x}{\te})}{\val}$, where $\x$ is
a dynamic variable, reduces to  a \emph{dynamic let}
$\DLet{\x}{\val}{\te}$.  In this construct, free occurrences of $\x$
in $\te$ are not immediately replaced by $\val$, as in the standard
static let, but rather reduction of $\te$ is started. When, during
this reduction, an occurrence of $\x$ is found in redex position, it
is replaced by the value of $\x$ in the innermost enclosing
\texttt{dlet}, so that dynamic scoping is obtained.

In our calculus, the behaviour of the dynamic let is obtained by the
unbind and rebind constructs. However, there are at least two
important differences. Firstly, the unbind construct allows the
programmer to explicitly control the program portions where a
variable should be dynamically bound. In particular, occurrences of
the same variable can be bound either statically or dynamically,
whereas {Moreau} \cite{M98} {assumes} two distinct sets.
Secondly, our rebind behaves in a hierarchical way, whereas, taking
{Moreau's} approach \cite{M98} where the innermost binding
is selected, a new rebind for the same variable would rewrite the
previous one, as also in {work by Dezani et al.} \cite{DGN08b}. For instance,
$\Rebind{\Rebind{\Unbind{\x}{\x}}{\x\mapsto 1}}{\x\mapsto 2}$ would
reduce to $2$ rather than to $1$. The advantage of our semantics, at
the price of a more complicated type system, is again more control.
In other words, when  the programmers want to use ``open code'',
they must explicitly specify the desired {binding}, whereas in
{Moreau's paper} \cite{M98} code containing dynamic variables is
automatically rebound with the binding which accidentally exists
when it is used. This semantics, when desired, can be recovered in
our calculi by using rebinds of the shape
$\Rebind{\te}{\x_1\mapsto\x_1,\ldots,\x_n\mapsto\x_n}$, where
$\x_1,\ldots,\x_n$ are all the dynamic variables which occur in
$\te$.

Other calculi for dynamic binding and/or rebinding have been
proposed {\cite{Dami97alambda-calculus,
KSS06, BiermanEtAl03a}}. We refer to {our previous papers introducing the calculus}
\cite{DezaniEtAl09,DezaniEtAl10} for a discussion and comparison.

As already mentioned, an interesting feature of our calculus is
that elements of the same set can play the double role of
\emph{standard variables}, which can be $\alpha$-renamed, and
\emph{names}, which cannot be $\alpha$-renamed (if not globally in a
program) {\cite{AnconaMoggi04, NanevskiPfenning05}}.
The crucial difference is that in the case of standard variables
the matching betweeen parameter and argument is done on a
\emph{positional} basis, as demonstrated by the de Bruijn notation,
whereas in the case of names it is done on a \emph{nominal} basis.
An analogous difference holds between tuples and records, and
between positional and name-based parameter passing in languages, as
recently discussed  {by Rytz and Odersky} \cite{RytzOdersky10}. \EZComm{La notazione a
record ha avuto successo nell'object-oriented, mentre nei linguaggi
funzionali si usano tuple per funzioni non curried, ma nei linguaggi
funzionali reali comunque poi ci sono sempre anche i record. La
notazione posizionale \`e vantaggiosa per chi scrive il codice
perch\'e non \`e vincolato a nomi particolari. Per il cliente invece
\`e migliore la notazione nominale perch\`e in genere i nomi sono
significativi. In ogni caso comunque il cliente e il fornitore
devono ``mettersi d'accordo'' su una convenzione, che sia di nome o
di posizione.}

 Distributed process calculi provide rebinding of names, see for
instance {the work of Sewell} \cite{Sewell2007}. Moreover, rebinding for
distributed calculi has been studied {\cite{AFG07}},
where, however, the problem of integrating rebinding with
standard computation is not addressed, so there is no interaction
between static and dynamic binding.

Finally, an important source of inspiration has been multi-stage
programming as, e.g., in {MetaML} \cite{MetaML}, notably for
the idea of allowing (open) code as a special value, the
hierarchical nature of the {unbind/rebind} mechanism and,
correspondingly, of the type system. The type system of Taha and
Sheard \cite{MetaML} is more expressive than the present one, since
both the turn-style and the types are decorated with integers. A
deeper comparison will be subject of further work.

In order to model different behaviours according
to the presence (and type concordance) of variables in the
rebinding environment, we plan to add a construct for conditional execution of rebind \MD{\cite{DGN08b}}. With this construct  we could model a
variety of object models, paradigms and language features.

Future investigation will also deal with the general form of
binding discussed by Tanter \cite{T09}, which subsumes both static and
dynamic binding and also allows fine-grained bindings which can
depend on contexts and environments.

\textbf{Acknowledgments.} We warmly thank the anonymous referees
for their useful comments. In particular, one referee {warned us about}
the problem of avoiding variable capture when applying substitution
to an unbound term. We also thank Davide Ancona for {pointing out the
work by Rytz and Odersky} \cite{RytzOdersky10} and  the analogy among the pairs
variable/name, tuple/record, positional/nominal, any
misinterpretation is, of course, our responsibility. 